\documentclass[12pt,DIV10,final]{scrartcl}

\usepackage{amsmath}
\usepackage{amssymb}
\usepackage{amsthm}
\usepackage{enumerate}
\usepackage{enumitem}
\usepackage[english]{babel}
\usepackage{ifthen}
\usepackage{multirow}
\usepackage{xspace}
\usepackage{url}
\usepackage{wasysym}

\usepackage[final]{hyperref}

\usepackage{tikz}

\definecolor{darkgreen}{rgb}{0.0,0.7,0.0}
\newenvironment{bs}{\noindent\color{darkgreen} BS:}{}

\newenvironment{mk}{\noindent\color{blue} MK:} {}

\newenvironment{vd}{\noindent\color{red} VD:} {}

\newcommand{\refthm}[1]{Theorem~\ref{#1}\xspace}

\newcommand{\reflem}[1]{Lemma~\ref{#1}\xspace}
\newcommand{\refprop}[1]{Proposition~\ref{#1}\xspace}

\newcommand{\refcon}[1]{Conjecture~\ref{#1}\xspace}

\newcommand{\refsec}[1]{Section~\ref{#1}\xspace}

\newcommand{\refex}[1]{Example~\ref{#1}\xspace}

\newcommand{\IFF}{if and only if\xspace}
\newcommand{\homo}{homomorphism\xspace}
\newcommand{\homos}{homomorphisms\xspace}
\newcommand{\schuetz}{Sch\"utz\-en\-ber\-ger\xspace}
\newcommand{\EF}{Ehren\-feucht-Fra{\"{\i}}ss{\'e}\xspace}
\newcommand{\CR}{Church-Rosser\xspace}

\newcommand{\swr}{sub\-word-redu\-cing\xspace}
\newcommand{\lr}{length-redu\-cing\xspace}

\newcommand{\set}[2]{\left\{#1\mathrel{\left|\vphantom{#1}\vphantom{#2}\right.}#2\right\}}
\newcommand{\oneset}[1]{\left\{\mathinner{#1}\right\}}
\newcommand{\smallset}[1]{\left\{#1\right\}}

\let\implies=\undefined

\newcommand{\implies}   {\text{$\;\Rightarrow\;$}}

\newcommand{\abs}[1]{\left|\mathinner{#1}\right|}

\newcommand{\N}{\mathbb{N}}

\newcommand{\logicfont}[1]{\mathrm{#1}}

\newcommand{\FO}{\logicfont{FO}}

\newcommand{\LTL}{\logicfont{LTL}}  

\newcommand{\wh}[1]{\widehat{#1}\,}

\newcommand{\sse}{\subseteq} 
\newcommand{\sm}{\setminus}

\newcommand{\IRR}{\mathrm{IRR}}

\renewcommand{\phi}{\varphi}
\newcommand{\eps}{\varepsilon}

\newcommand{\alp}{\alpha}

\newcommand{\gam}{\gamma}

\newcommand{\sig}{\sigma}

\newcommand{\Gam}{C}

\newcommand\RAS[2]{\overset{#1}{\underset{#2}{\Longrightarrow}}}

\newcommand\ra[1]{\overset{#1}{\longrightarrow}}
\newcommand\LAS[2]{\overset{#1}{\underset{#2}{\Longleftarrow}}}
\newcommand\DAS[2]{\overset{#1}{\underset{#2}{\Longleftrightarrow}}}

\newcommand\RA[1]{\underset{#1}{\Longrightarrow}}

\newcommand{\SF}{\mathrm{SF}}
\newcommand{\AP}{\mathrm{AP}}

\newtheorem{theorem}{Theorem}
\newtheorem{definition}[theorem]{Definition}
\newtheorem{proposition}[theorem]{Proposition}
\newtheorem{lemma}[theorem]{Lemma}

\newtheorem{remark}[theorem]{Remark}
\newtheorem{conjecture}{Conjecture}

\newtheorem{expl}[theorem]{Example}
\newenvironment{example}[1][]{\ifthenelse{\equal{#1}{}}{\begin{expl}\upshape}{\begin{expl}[#1]\upshape}}{\end{expl}}

\newcommand{\eex}{\hspace*{\fill}\ensuremath{\Diamond}}

\setlist{itemsep=3pt,parsep=3pt,topsep=4pt}

\newcommand{\dotcup}{\mathbin{\dot{\cup}}}

\begin{document}

\title{Star-Free Languages are Church-Rosser Congruential}
\author{%
  Volker Diekert\thanks{%
    Institut f\"ur Formale Methoden der Informatik, %
    University of Stuttgart. 
%
    The work on this paper has been initiated by the program
    \emph{Automata Theory and Applications} at the Institute for
    Mathematical Sciences, National University of Singapore in
    September~2011.  The first author would like to thank NUS for the
    hospitality and the organizing committee chaired by Frank Stephan
    for the invitation. He also thanks Klaus Reinhardt for the
    introduction to the topic.}
  \and Manfred Kuf\-leitner\thanks{%
    Institut f\"ur Formale Methoden der Informatik, %
    University of Stuttgart, Germany.
    The second author was supported by the German Research Foundation
    (DFG) under grant \mbox{DI 435/5-1}.}
  \and Pascal Weil\thanks{%
    LaBRI, Universit{\'e} de Bordeaux and CNRS, France. 
    The third author was supported by
    the grant ANR 2010 BLAN 0202 01 FREC.
  }}

\date{November 17, 2011}
\maketitle

\begin{abstract}
  \noindent
  \textbf{Abstract.} 
  The class of Church-Rosser congruential languages has been
  introduced by \mbox{McNaughton}, Narendran, and Otto in 1988.  A
  language $L$ is Church-Rosser congruential (belongs to CRCL), if
  there is a finite, confluent, and length-reducing semi-Thue system
  $S$ such that $L$ is a finite union of congruence classes modulo
  $S$. To date, it is still open whether every regular language is in
  CRCL.  In this paper, we show that every star-free language is in
  CRCL. In fact, we prove a stronger statement: For every star-free
  language $L$ there exists a finite, confluent, and subword-reducing
  semi-Thue system $S$ such that the total number of congruence
  classes modulo $S$ is finite and such that $L$ is a union of
  congruence classes modulo $S$.  The construction turns out to be
  effective.

  \medskip

\end{abstract}

\section{Introduction}\label{intro}

Church-Rosser congruential languages (CRCL) are a nonterminal-free
form of \CR languages (CRL).  Both classes have been defined
in~\cite{McNaughtonNO88}, and it was shown there that CRCL forms a
proper subclass in CRL. Languages in CRL enjoy various nice
properties. For example their word problem is decidable in linear
time.  A detailed discussion with links to further references can be
found in the PhD-thesis of Niemann~\cite{NiemannPhD02}, see also
\cite{NiemannO05}.  We content ourselves to define CRCL: A language $L
\in A^*$ is called a \emph{Church-Rosser congruential language}, if
there is a finite, \lr, and confluent semi-Thue system $S \sse A^*
\times A^*$ such that $L$ is a finite union of congruence classes
modulo $S$. This means that $L$ contains a finite set $F$ of shortest
words such that we have $w \in L$ \IFF every rewriting procedure
starting on $w$ and using $S$ terminates in one of the finitely many
words in $F$.

It was also shown in~\cite{McNaughtonNO88} that all deterministic
context-free languages are Church-Rosser.  However, surprisingly
it is not known whether all regular languages are CRCL.  The general
conjecture is ``yes'', but so far only partial results have been
established as in~\cite{NiemannW02}. The most advanced result has been
announced by Reinhardt and Th\'erien~\cite{reinhardtT03}: According to
their manuscript, if a regular language has a group as its syntactic
monoid, then this language is in CRCL.

In this note we consider the complementary class of group-free regular
languages; and we show that they belong to CRCL. A regular language is
\emph{group-free} if its syntactic monoid is group-free. This means it
is {\em aperiodic}. There are many other characterizations for this
class.  A fundamental result of \schuetz says that the class of
aperiodic language $\AP(A)$ is exactly the same as the class of
star-free languages $\SF(A)$~\cite{sch65sf:short}. It is the class
where the Krohn-Rhodes decomposition leads to a wreath product of the
three-element commutative idempotent reset-monoid
$U_2$~\cite{kr65tams}. It is also the class $\FO(A,<) $ of languages
definable in first-order logic~\cite{mp71:short}; and this is the same
as the class $\LTL(A)$ of languages definable in the linear temporal
logic~\cite{kam68}.

A proof that $\FO(A,<) = \SF(A) = \AP(A) = \LTL(A)$ can be
conveniently arranged in a cycle. The inclusion $\FO(A,<) \sse
\SF(A)$ can be explained very nicely with \EF-games
\cite{Ehrenfeucht61}. The inclusion $\SF(A) \sse \AP(A)$ follows
\schuetz's original idea. The inclusion $\AP(A) \sse \LTL(A)$ is
done in the survey~\cite{dg08SIWT:short} with the concept of
\emph{local divisors }which play a prominent role here, too.  The
final inclusion $\LTL(A) \sse \FO(A,<)$ is trivial.

Coming back to the class of Church-Rosser congruential languages, our
main result shows $\SF(A) \sse $ CRCL. Actually, we prove a much
stronger result.  First we define \swr semi-Thue systems which are a
proper subclass of finite \lr semi-Thue systems. For every language $L
\in \AP(A)$ we effectively construct a finite \swr confluent semi-Thue
system $ S \sse A^* \times A^*$ such that the total number of
congruence classes modulo $S$ is finite and $L$ is a union of such
classes, see \refthm{thm:main}. A main tool in our proof is the notion
of \emph{local divisor}, see \refsec{sec:ld} for a definition.

\enlargethispage{\baselineskip}

In the final section of this paper, \refsec{sec:alg}, we explain our
constructions in a rigorously algebraic framework. This part is
mainly intended for possible future work.

In order to give a complete positive solution to the conjecture that
all regular languages are CRCL, it remains to combine our approach
with the one in~\cite{reinhardtT03}. There are however quite a number
of obstacles for a fruitful combination. So, we leave the general
conjecture as a challenging research problem.

\section{Preliminaries and Notation}\label{pre}

In the following $A$ means a finite alphabet, an element of $A$ is
called a {\em letter}, and $A^*$ denotes the {\em free monoid}
generated by $A$. It is the set of {\em words} over $A$. The empty
word is denoted by 1. The {\em length} of a word $u$ is denoted by
$\abs u$. We have $\abs u = n$ for $u= a_1 \cdots a_n$ where $a_i \in
A$. The empty word has length $0$.  We carefully distinguish between
the notion of factor and subword.  Let $u,v \in A^*$. The word $u$ is
called a {\em factor} of $v$ if there is a factorization $v= xuy$. It
is called a {\em subword} of $v$ if there is a factorization $v=
x_0u_1x_1 \cdots u_kx_k$ such that $u= u_1 \cdots u_k$.  A subword is
also sometimes called a {\em scattered subword} in the literature.

A {\em semi-Thue system} over $A$ is a subset $S\sse A^*\times
A^*$. The elements are called {\em rules}. We frequently write
$\ell \ra{} r$ for rules $(\ell,r)$.  A system $S$ is called
{\em \lr} if we have $\abs \ell > \abs r $ for all rules $(\ell,r)
\in S$. It is called {\em \swr,} if $r$ is a subword of $\ell$ and
$\ell \neq r$ for all rules $(\ell,r) \in S$. Every \swr system is
\lr, but not vice versa.

Every system $S$ defines the rewriting relation ${\RA S} \sse A^*
\times A^*$ by
\begin{align*}
  u \RA S v \;\text{ if } \; u=p\ell q, \; v= prq \; \text{ for some
    rule } \; (\ell,r) \in S.
\end{align*}

By $\RAS*{S}$ we mean the reflexive and transitive closure of
$\RA{S}$.  By $\DAS*{S}$ we mean the symmetric, reflexive, and
transitive closure of $\RA{S}$. We also write $u \LAS*{S}v$ whenever
$v\RAS*{S}u$.  The system $S$ is {\em confluent} if for all
$u\DAS*{S}v$ there is some $w$ such that $u\RAS*{S}w\LAS*{S}v$.

Note that $u \RAS{}S v$ implies that $\abs u > \abs v$ for \lr
systems. For \swr systems it implies that the set of subwords in $v$
is a proper subset of the set of subwords in $u$.

By $\IRR(S)$ we denote the set of irreducible words, i.e., the set of
words where no left-hand side occurs as any factor.  The relation
${\DAS * S} \sse A^* \times A^*$ is a congruence, hence the congruence
classes $[u]_S = \{v \in A^*\mid u \DAS*S v\}$ form a monoid which is
denoted by $A^*/S $. A finite semi-Thue system $S$ can be viewed as a
finite set of defining relations. Hence, $A^*/S $ becomes a finitely
presented monoid.

\begin{definition} 
  A semi-Thue system $S$ is called a {\em Church-Rosser system} if it
  is \lr and confluent.  A language $L \sse A^*$ is called a {\em
    Church-Rosser congruential language} if there is a finite
  Church-Rosser system $S$ such that $L$ can be written as a finite
  union of congruence classes $[u]_S$.
\end{definition}

\begin{remark}\label{rem:basic}
  A semi-Thue system $S$ is a Church-Rosser system \IFF (1) it is
  length-reducing and (2) every congruence class has exactly one
  irreducible element.

  Let $\pi: A^* \to A^*/S, \; u \mapsto [u]_S$ be the canonical \homo
  and $S$ be a finite Church-Rosser system. Then $\pi^{-1}(K)$ is a
  Church-Rosser congruential language as soon as $K$ is finite.
\end{remark}
 
 
\begin{conjecture}\label{con:mcnno}
  Every regular language is a Church-Rosser congruential language.
\end{conjecture}

\pagebreak[1]

\begin{example}\label{exa:1}
  Consider the language $L = (bc)^+$. A Church-Rosser system for $L$
  is given by the one-rule semi-Thue system $S = \oneset{cbc
    \longrightarrow c}$. The monoid $
  \smallset{b,c}^*/S$ is infinite. However $L = [bc]_S$; and hence $u
  \in L$ if and only if $u \RAS*{S} bc$.
  \eex
\end{example}

A manuscript of Reinhardt and Th\'erien~\cite{reinhardtT03} says that
\refcon{con:mcnno} is true in case that the syntactic monoid of the
regular language is a group.  Here, we are going to prove an even
stronger result for aperiodic languages, i.e., for languages where the
syntactic monoid is group-free.  As our proof uses subword-reducing
systems in the induction hypothesis, we cannot incorporate the
statement of Reinhardt and Th\'erien (using length-reducing rather
than subword-reducing systems) as base in our induction scheme.  So
the status of \refcon{con:mcnno} remains open, in general.

\begin{definition}\label{def:cr}
  Let $\phi: A^* \to M$ be a \homo to a finite monoid $M$.  We say
  $\phi$ {\em factorizes through a finite Church-Rosser monoid
    $A^*/ S$} if there is a finite Church-Rosser system $S$ such
  that $A^*/ S$ is a finite monoid and $[u]_S \sse
  \phi^{-1}(\phi(u))$ for all $u \in A^*$.
\end{definition}

A classical fact states that a language $L \sse A^*$ is regular
\IFF it is {\em recognizable}, i.e., there is a \homo $\phi: A^*
\to M$ to a finite monoid $M$ such that $L = \phi^{-1}(\phi(L))$. We
also say that $\phi$ (or that $M$) recognizes $L$.  Recall that a
finite monoid $M$ is called {\em aperiodic} if there exists some $n\in
\N$ such that $x^n = x^{n+1}$ for all $x \in M$. Accordingly, a
language $L \sse A^*$ is called {\em aperiodic} if it is recognized
by some finite aperiodic monoid $M$.

Note that if $\phi$ factorizes through a finite Church-Rosser monoid,
then we have $$\phi: A^* \overset{\pi}\to A^*/S
\overset{\psi}\to M,$$ where $S$ is a Church-Rosser system such that
$A^*/ S$ is a finite.

\begin{conjecture}\label{con}
  Let $\phi: A^* \to M$ be a \homo to a finite monoid $M$.  Then
  $\phi$ factorizes through a finite Church-Rosser monoid.
\end{conjecture}

\refcon{con} is stronger than \refcon{con:mcnno}. However, we believe
that a positive solution to \refcon{con:mcnno} comes through a proof
of \refcon{con}.  Actually, the result in~\cite{reinhardtT03} also
announces that \refcon{con} is true for finite groups. We are going to
show here that an even stronger statement than \refcon{con} holds for
finite aperiodic monoids.

\begin{example}\label{exa:2}
  Consider again the language $L = (bc)^+$ from \refex{exa:1}. Another
  Church-Rosser system for $L$ is given by 
  \begin{alignat*}{2}
    S = \{ 
      && bbb &\longrightarrow bb,\
      bbc \longrightarrow bb,\ 
      cbb \longrightarrow bb, \\
      && ccc &\longrightarrow bb,\ 
      ccb \longrightarrow bb,\ 
      bcc \longrightarrow bb \\
      && bcb &\longrightarrow b, \
      cbc \longrightarrow c
    \}.
  \end{alignat*}
  As in \refex{exa:1} we have $L = [bc]_S$; but here, the monoid
  $\smallset{b,c}^*/S$ is finite. It has $7$ elements: $[1]_S$,
  $[b]_S$, $[c]_S$, $[bc]_s$, $[cb]_s$, $[bb]_S$, and $[cc]_S$.
  Note that $S$ is not subword-reducing.
  \eex
\end{example}

\section{Local divisors}\label{sec:ld}

The notion of {\em local divisor} dates back to a technical report of
Meyberg where he introduced this concept in commutative algebra,
see~\cite{FeTo02,Mey72}.  In finite semigroup theory and formal
language theory the explicit definition of a local divisor appeared
first in~\cite{dg06IC}. Since then it turned out to be a very useful
tool for simplifying classical proofs like
in~\cite{dg08SIWT:short,DiekertKS11} or in finding new results like in
this paper. The definition of a local divisor extends the definition
of a Sch\"utzenberger group for the $\mathcal H$-class of an arbitrary
element,~\cite{cp67,rs09qtheory}.  A category generalization is being
used by Steinberg and Costa in the context of symbolic dynamics
(unpublished).

In this paper we use local divisors for aperiodic monoids, only.  Let
$M$ be a monoid and let $c \in M$.  We put on the subsemigroup $cM
\cap Mc$ a monoid structure by defining a new multiplication $\circ$
as follows:
$$xc \circ cy = xcy.$$
It is straightforward to see that $\circ$ is well-defined and $(cM
\cap Mc, \circ)$ is a monoid with neutral element $c$. 
  
The following observation is crucial: If the monoid $M$ is finite and
aperiodic, then $\abs{cM \cap Mc} < \abs{M}$ whenever $c \neq 1$. This
is clear, because $1 \in {cM \cap Mc}$ implies that $c$ is a unit of
$M$, but $c \neq 1$ and there are no non-trivial units in aperiodic
monoids.  The set $M' = \set{x}{cx \in Mc}$ is a submonoid of $M$, and
$c{\cdot}: M' \to cM \cap Mc : x \mapsto cx$ is a surjective
homomorphism. In particular, if $M$ is aperiodic, then $(cM \cap Mc,
\circ)$ is aperiodic, too. Since $(cM \cap Mc, \circ)$ is the
homomorphic image of a submonoid it is a divisor of $M$.  We therefore
call $(cM \cap Mc, \circ)$ the {\em local divisor} of $M$ at $c$. Note
that if $c = c^2$ is an idempotent, then $(cM \cap Mc, \circ) =
(cMc,\cdot)$ is the usual {\em local monoid} defined by the
subsemigroup $cMc$ of $M$. Thus, the notion of local divisor
generalizes the notion of local monoid from idempotents to arbitrary
elements.

\section{\refcon{con} holds for aperiodic monoids}\label{sec:con}
 We have the following result.

\begin{theorem}\label{thm:main}
  Let $\phi: A^* \to M$ be a \homo to a finite aperiodic monoid $M$.
  Then $\phi$ factorizes through a finite aperiodic Church-Rosser
  monoid $A^* / S$ where $S$ is subword-reducing.
\end{theorem}

The rest of this section is devoted to the proof of \refthm{thm:main}. 
The proof is by induction on the parameter $(\abs{M},\abs{A})$ with
lexicographic order.  The result is true if $\varphi(A^*)$ is
trivial. Note that this covers $M = \smallset{1}$ as well as $A =
\emptyset$.  In the remaining case there is a letter $c \in A$ such
that $\phi(c) \neq 1$. We let $B= A\sm \smallset{c}$, and for better
reading we identify $c$ and $\phi(c)\in M$.  Since $c \neq 1 \in M$
and $M$ is aperiodic, $c$ is not a unit.  Hence $M_c= cM \cap Mc$ has
less elements than $M$.
  
Since $\abs B < \abs A$ we find, by induction, a finite \swr
Church-Rosser system $R \sse B^* \times B^*$ such that the restriction
$\phi|_{B^*}: B^*\to M$ factorizes through a finite Church-Rosser
monoid $B^*/ R$. In particular, $(\ell, r ) \in R$ implies $\phi(\ell)
= \phi(r)$.

For $u \in B^*$ let $\wh {u}$ denote the unique word such that $\wh
{u}\in \IRR(R)$ and $u\RAS*R \wh {u}$. The subset $K = \IRR(R)c \sse
A^*$ is a finite code.  This means that $K^*$ is freely generated, as
a submonoid of $A^*$, by the finite set $K$. Note that $K^+ \sse
A^*c$.  Consider the \homo $\psi: K^*\to (M_c, \circ)$ which is given
by $\psi(\wh u c) = c\phi(u)c$. We have $c\phi(u)c = \phi(c \wh u c)$.
In particular, $\psi$ is well-defined. By induction $\psi: K^*\to
(M_c, \circ)$ factorizes through a finite aperiodic Church-Rosser
monoid $K^*/ T$, where $T\sse K^* \times K^*$ is a finite \swr
Church-Rosser system.
  
Consider a rule $(\ell, r) \in T$. It has the form
$$ \wh{u_1}c  \cdots \wh{u_m}c \ra {}\wh{v_1}c \cdots \wh{v_n}c$$
where the $\wh{u_i}c$ and $\wh{v_j}c$ are letters in $K$, every
right-hand side $\wh{v_1}c \cdots \wh{v_n}c\in K^*$ is a proper
subword of $\wh{u_1}c \cdots \wh{u_m}c\in K^+$.  Since
$K^* \sse A^*$ we can read $T$ as a semi-Thue system over $A$ as
well. Next, we define a new system $\wh{T} \sse A^* \times A^*$ as
follows:
$$\wh{T} = \set{c\ell \ra{} c{r}}{(\ell, r) \in T}.$$
   
We collect some important properties of $\wh{T}$ in a remark: 

\begin{remark}\label{rem:wht}
  The semi-Thue system $\wh{T}\sse A^* \times A^*$ 
  satisfies the following assertions.
  \begin{enumerate}
  \item $\wh{T}$ is \swr, because $T$ has this property. This is
    crucial.  Knowing only that $T$ is \lr as a system over $K^*$
    would not be enough to conclude that $\wh{T}$ is \lr as a system
    over $A^*$.
  \item $\wh{T}$ is confluent. For this it is crucial that we added a
    letter $c$ on the left. This allows to read the words $\wh u c$ as
    letters in $K$ and the confluence of $T$ transfers to the
    confluence of $\wh{T}$. If there was no $c$ on the left, then $T$
    could contain rules $abc \ra{} 1$ and $bc \ra{} 1$, but $a$ is no
    left-hand side in $T$. Over $K$ the words $abc$ and $bc$ are
    letters, hence there is no overlap in $K^*$.
  \item $c\ell \ra{} c{r} \in \wh {T}$ implies $\phi(c\ell) =
    \phi(cr)$.  This is a straightforward calculation in local
    divisors: Let $c\ell = c u_1 c \cdots u_m c$ and $cr = c v_1 c
    \cdots v_n c$ with $u_i, v_i \in \IRR(R)$. By induction, we have
    $\psi(\ell) = \psi(r)$ and thus
    \begin{align*}
      \varphi(c \ell) 
      &= \varphi(cu_1c) \circ \cdots \circ \varphi(cu_mc) \\
      &= \psi(u_1c) \circ \cdots \circ \psi(u_mc) \\
      &= \psi(u_1c \cdots u_m c) 
      = \psi(v_1c \cdots v_n c) \\
      &= \psi(v_1c) \circ \cdots \circ \psi(v_n c) \\
      &= \varphi(cv_1c) \circ \cdots \circ \varphi(cv_nc) = \varphi(c r).
    \end{align*}
  \end{enumerate}
\end{remark}

The proof of \refthm{thm:main} is now a direct consequence of the
following lemma which shows that the system $S = R \cup \wh T$ has the
desired properties.

\begin{lemma}\label{lem:main}
  The semi-Thue system $S = R \cup \wh T$ over $A$ satisfies the
  following assertions.
  \begin{enumerate}
  \item\label{i} $S$ is \swr.
  \item\label{ii} $S$ is confluent.
  \item\label{iii} $\ell \ra{} {r} \in S $ implies $\phi(\ell) =
    \phi(r)$.
  \item\label{iv} $A^*/S$ is a finite aperiodic monoid.
  \end{enumerate}
\end{lemma}

\begin{proof}
  Assertion \ref{i} is clear, because $R$ and $\wh T$ are \swr.
  Assertion \ref{ii} is clear, because there is no overlap of
  left-hand sides between rules of $R$ and $\wh T$.  Assertion
  \ref{iii} is clear, because $R$ and $\wh T$ have this property.  It
  remains to show \ref{iv}. By induction $K^*/ T$ is finite. Hence
  there is a maximal value $\mu$ such that every word in $K^*$ of
  length at least $\mu$ is reducible. We conclude that: 
  \begin{equation*}
    \IRR(S) \sse \set{\wh {u_0} c \wh {u_1}\cdots c\wh{u_m}} 
    { \wh {u_i} \in \IRR(R) \wedge 0 \leq m \leq \mu}.
  \end{equation*}
  Since $\IRR(R)$ is finite, we see that $\IRR(S)$ is a subset of a
  finite set, and thus the finiteness of $\IRR(S)$ and of $A^*/S$
  follow. This leaves us to show that $A^*/S$ is aperiodic.  We have
  to show that there exists some $n\in \N$ such that for all $u = \wh
  {u_0}c \wh {u_1}\cdots c \wh{u_m} \in \IRR(S)$ we have $u^{n+1}
  \DAS*S u^{n}$. Let $v = \wh {u_1}c\cdots \wh {u_m\, u_0}c$. Then
  $u^{n+1} \DAS*R p c v^n q$ and $u^n \DAS*R p c v^{n-1} q$ for some
  $p,q \in A^*$.  Therefore, it is enough to show that $cv^{n} \DAS*S
  cv^{n-1}$ whenever $n$ is large enough.  The $\wh {u_i}c$'s are code
  words of $K$, hence letters in the alphabet $K$ and we can read $v
  \in K^*$. Here we can use induction, and we know $v^{n} \DAS*T
  v^{n-1}$ if $n$ is large enough, because $K^*/ T$ is aperiodic.
  This implies $cv^{n} \DAS*{\wh{T}} cv^{n-1}$ and hence the result.
\end{proof}

This completes the proof of \refthm{thm:main}.

\begin{example}
  Consider again the language $L = (bc)^+$ from \refex{exa:1} and
  \refex{exa:2}. Its syntactic monoid is $M = \oneset{1,b,c,bc,cb,0}$
  with $bb = cc = 0$, $bcb = b$, $cbc = c$, $1$ is neutral, and $0$ is
  a zero element. In particular, $bc$ and $cb$ are idempotent. Here,
  the syntactic homomorphism $\varphi_L : \oneset{b,c}^* \to M$ is
  induced by $b \mapsto b$ and $c \mapsto c$.  We apply the above
  algorithm for obtaining a \CR monoid factorizing $\varphi_L$.

  First we choose to localize at $c$. Then $N = \oneset{1,b,0}$ is the
  submonoid generated by $b$. The restriction of $\varphi_L$ to $b^*$
  factorizes through the \CR monoid defined by the system 
  \begin{equation*}
    R = \oneset{bbb \longrightarrow bb}.
  \end{equation*}
  This leads to the irreducible elements $\IRR(R) =
  \oneset{1,b,bb}$. Now, the homomorphism $\psi : \oneset{c,bc,bbc}^*
  \to M_c$ is defined by $x \mapsto cx$ for $x \in
  \oneset{c,bc,bbc}$. Note that we consider $\oneset{c,bc,bbc}$ as a
  three-letter alphabet. In particular, $M_c = \oneset{c,0}$ and $c
  \mapsto 0$, $bbc \mapsto 0$, and $bc \mapsto c$.

  For $\psi$ we obtain the rules 
  \begin{alignat*}{2}
    T = \{\
    && (c)(c) &\longrightarrow (c), \\
    && (bc) &\longrightarrow 1,\\
    && (bbc)(bbc) &\longrightarrow (bbc), \\ 
    && (c)(bbc)(c) &\longrightarrow (c)
    \qquad \}
  \end{alignat*}
  The parenthesis are for identifying letters of the alphabet of
  $\psi$. This leads to the system
  \begin{alignat*}{2}
    \widehat{T} = \{\
    && ccc &\longrightarrow cc, \\
    && cbc &\longrightarrow c,\\
    && cbbcbbc &\longrightarrow cbbc, \\
    && ccbbcc &\longrightarrow cc
    \qquad \}
  \end{alignat*}
  and $S = R \cup \widehat{T}$ is the system for $\varphi$. In
  $\IRR(S)$ there are $65$ irreducible elements and $bbcbbccbbcbb$ is
  the longest one.
  \eex
\end{example}

\section{Algebraic constructions}\label{sec:alg}

The aim of this section is to place the explicit constructions from the 
previous \refsec{sec:con} into a broader algebraic context.
It shows that the quotient monoid $A^*/S$ in \reflem{lem:main}
has an algebraic interpretation.

\subsection{Rees-extension monoids and Church-Rosser systems}

Let $\rho: P\to Q$ be a mapping between two monoids $P$ and
$Q$. 
We are going to define the \emph{Rees-extension monoid} of $\rho$
which we shall denote by $E(\rho)$.
If $\rho$ is chosen properly, then $E(\rho)$ coincides with the monoid
$A^*/S$ where $S \sse A^* \times A^*$ is the \swr confluent semi-Thue
system of \reflem{lem:main}, see \refprop{WasSollDieseSection}.
As a carrier set for the monoid $E(\rho)$ we choose the disjoint union
$P \dotcup (P\times Q \times P)$. The multiplication is as follows:
\begin{alignat*}{2}
  u \cdot v &= uv   
  &\quad&\text{for $u,v\in P$.}   \\
  x \cdot (u,q, v)\cdot y &=(xu,q, vy)  
  &\quad&\text{for $x,u,v,y\in P$ and  $q \in Q$.}\\
  (u,q, v)\cdot(x,r, y)&=(u,q\,\rho(vx)\,r,y )  
  &\quad&\text{for $u,v,x,y\in P$ and  $q,r \in P$.}
\end{alignat*}
Now, $P$ is a submonoid of $E(\rho)$ and $P \times Q \times P$ is an
ideal.  As a semigroup, $P \times Q \times P$ is a special case of the
Rees-matrix construction, see e.g.~\cite{cp67,rs09qtheory}: The
mapping~$\rho$ defines a $P \times P$ matrix $\mathcal{R}$ with
coefficients in $Q$ by $\mathcal{R}(v,x) = \rho(vx)$; and the
multiplication in $P \times Q \times P$ can be written as $(u,p,
v)\cdot(x,q, y) = (u,p\,\mathcal{R}(v,x)\,q,y)$.

In the following we let $c = \rho(1)\in Q$.  Multiplying triples
$(1,q,1)$ and $(1,r,1)$ yields $(1,q,1)\cdot (1,r,1) = (1,q\,\rho(1)\,
r,1)= (1,qcr,1)$. In particular, the sandwich construction $(Q,\#_c)$
appears as a subsemigroup, where $\#_c$ denotes the standard
sandwich-multiplication defined by $q\mathop{\#_c}r = qcr$.  We have
$(1,1,1)^n = (1,c^{n-1},1)$ and, more general, $(u,q,v)^n= \big( u,
\big(q\, \rho(vu)\big)^{n-1}q, v\big)$ for all $n \geq 1$.  It follows
that $E(\rho)= P \dotcup (P\times Q \times P)$ is aperiodic if both
$P$ and $Q$ are aperiodic.

For \refprop{WasSollDieseSection} below, we apply the
Rees-extension monoid to the setting in \refsec{sec:con}. We start
with a \homo $\phi: A^* \to M$ to a finite aperiodic monoid $M$, the
alphabet $A$ is the disjoint union of $B$ and $\smallset{c}$, $P$ is
the quotient $B^* / R$, $Q$ is the quotient $K^* / T$ for $K = \IRR(R)
c$. Since we can identify $P = B^* / R$ and $\IRR(R)$, we define $\rho
: P \to Q$ by $\rho(u) = [uc]_T$ for $u \in \IRR(R)$. Now, $A^* / S$
from \refsec{sec:con} and $E(\rho)$ coincide:

\begin{proposition}\label{WasSollDieseSection}
  In the situation above, $A^* / S$ and $E(\rho)$ are
  isomorphic.
\end{proposition}

\begin{proof}
  Let $\sigma : \IRR(S) \to E(\rho)$ be defined by 
  \begin{align*}
    \sigma(u_0) &= u_0 \qquad \text{and} \\
    \sigma(u_0 c u_1 \cdots c u_{k+1}) &=
    \big(u_0,\rho(u_1) \cdots \rho(u_{k}), u_{k+1}\big)
  \end{align*}
  for $k \geq 0$ and $u_i \in B^* \cap \IRR(R)$. Here, we indentify
  $P$ with $\IRR(R)$, and $Q$ with $\IRR(T)$.  In particular, by
  definition of $P$ and $Q$, the mapping $\sigma$ is
  surjective. Suppose $\sigma(u_0 c u_1 \cdots c u_{k+1}) = \sigma(v_0
  c v_1 \cdots c v_{\ell+1})$ for $k,\ell \geq 0$ and $u_i, v_i \in
  B^* \cap \IRR(R)$. Then $u_0 = v_0$ and $u_{k+1} = v_{\ell +
    1}$. Moreover, $c u_1 c \cdots u_k c \in \IRR(S)$ and thus $(u_1c)
  \cdots (u_k c) \in K^* \cap \IRR(T)$.  Similarly, $(v_1c) \cdots
  (v_\ell c) \in K^* \cap \IRR(T)$.  Now, $\rho(u_1) \cdots \rho(u_k)
  = \rho(v_1) \cdots \rho(v_\ell)$ implies $(u_1c) \cdots (u_k c) =
  (v_1 c) \cdots (v_\ell c)$ in $K^*$ and thus $cu_1c \cdots u_k c =
  cv_1c \cdots v_\ell c$ in $A^*$.  This shows $u_0 c u_1 \cdots c
  u_{k+1} = v_0 c v_1 \cdots c v_{\ell + 1}$. We conclude that
  $\sigma$ is injective.

  It remains to show that $\sigma$ is a homomorphism. Let $u, v \in
  \IRR(S)$ and $uv \RAS*S w \in \IRR(S)$, i.e., $[u]_S [v]_S =
  [w]_S$. If $u,v \in B^*$, then $\sigma(u)\sigma(v) = w = \sigma(w)$.
  Let now $uv,w \in A^* c A^*$ and $w = w_0 c w_1 \cdots c w_{m+1}$.
  If $u \in B^*$, $v = v_0 c v_1 \cdots c v_{\ell+1}$ and $uv_0 \RAS*S
  x \in \IRR(S)$, then $w_0 = x$ and $c w_1 \cdots c w_{m+1} = c v_1
  \cdots c v_{\ell + 1}$. It follows $\sigma(u)\sigma(v) =
  (x,\rho(w_1) \cdots \rho(w_m),w_{m+1}) = \sigma(w)$. The case $v \in
  B^*$ is symmetric.

  Let now $u = u_0 c u_1 \cdots c u_{k+1} \in A^* c A^*$ and $v = v_0
  c v_1 \cdots c v_{\ell+1} \in A^* c A^*$ with $u_i, v_i \in B^*$.
  Let $u_{k+1} v_0 \RAS*S x \in \IRR(S)$. Then $u_{k+1} v_0 \RAS*R x$.
  We have $c u_1 \cdots c u_{k+1} v_0 c \cdots v_k c \RAS*S c w_1 c
  \cdots w_m c$. By construction of $S$ we see that 
  \begin{equation*}
    c u_1 \cdots c
    u_k c x c v_1 c \cdots v_\ell c \RAS*{\wh T} c w_1 c \cdots w_m c,
  \end{equation*}
  and hence
  \begin{equation*}
    (u_1 c) \cdots (u_{k} c) (x c) (v_1 c) \cdots (v_\ell c)
    \RAS*T (w_1 c) \cdots (w_m c),
  \end{equation*}
  i.e., $\rho(u_1) \cdots \rho(u_k) \rho(x) \rho(v_1) \cdots
  \rho(v_\ell) = \rho(w_1) \cdots \rho(w_m)$ in $Q$.  We conclude that
  $\sigma(u) \sigma(v) = \sigma(w)$.
\end{proof}

\subsection{Rees-extension monoids and local divisors}

Let $\rho : P \to Q$ be arbitrary again.  Observe that $c\neq 1
\in Q$, in general. In the remainder of this section, we draw a
connection between local divisors and the Rees-extension monoid.  We
define an alphabet $\Gam$ by the disjoint union $\Gam = (P \setminus
\smallset{1}) \dotcup \smallset{c}$.  The mapping $\rho$ induces a
homomorphism $\tau : \Gam^* \to E(\rho)$ by defining $\tau(x) = x$
for $x \in P \setminus \smallset{1}$ and $\tau(c) = (1,1,1)$. By
considering $(P\setminus \smallset{1})^* c$ as an infinite alphabet,
$\rho$ also induces a homomorphism $\sigma : \big((P\setminus
\smallset{1})^* c\big)^* \to Q$ by $\sigma(uc) = \rho(\eps(u))$ for $u
\in (P\setminus \smallset{1})^*$. Here, $\eps : (P\setminus
\smallset{1})^* \to P$ is the evaluation homomorphism.

Consider a \homo $\gam: \Gam^* \to M$ with $\gam(c) = c \in
M$.  The aim is to find a condition such that $\gam$ factorizes
through $\tau: \Gam^* \to E(\rho)$. This means we wish to write $\gam
= \tau \psi$ for some suitable \homo $\psi:E(\rho) \to M$. The
condition we are looking for is statement \ref{lifti} of
\refprop{prop:lift}.

\begin{proposition}\label{prop:lift}
  Let $\gam: \Gam^* \to M$ be a \homo with $\gam(c) = c \in M$. If $Q$
  is generated by $\rho(P)$, then the following assertions are
  equivalent.
  \begin{enumerate}
  \item\label{lifti} For $w, w'\in \big((P\setminus \smallset{1})^*
    c\big)^*$ the equality $\sigma(w) = \sigma(w')\in Q$ implies
    $c\gam(w) = c\gam( w')\in cM \cap Mc$.
  \item\label{liftii} There exists a \homo $\psi_c: Q \to M_c$ with
    $M_c = (cM \cap Mc, \circ)$ such that the following diagram commutes.
    \begin{center}
      \begin{tikzpicture}
        \draw (0,0) node (Pc) {$\big((P\setminus \smallset{1})^* c\big)^*$};
        \draw (3,0) node (Q) {$Q$};
        \draw (0,-2) node (M) {$Mc \cup \smallset{1}$};
        \draw (3,-2) node (Mc) {$M_c$};
        \draw[->] (Pc) -- node[above] {$\sigma$} (Q);
        \draw[->] (Pc) -- node[left] {$\gam$} (M);
        \draw[->] (Q) -- node[right] {$\psi_c$} (Mc);
        \draw[->] (M) -- node[above] {$c {\cdot}$} node[below] {$x \mapsto cx$} (Mc);
      \end{tikzpicture}
    \end{center}
  \item\label{liftiii} There exists a \homo $\psi: E(\rho) \to M$ such
    that the following diagram commutes.
    \begin{center}
      \begin{tikzpicture}
        \draw (0,0) node (C) {$C^*$};
        \draw (3,0) node (E) {$E(\rho)$};
        \draw (0,-2) node (M) {$M$};
        \draw[->] (C) -- node[above] {$\tau$} (E);
        \draw[->] (C) -- node[left] {$\gam$} (M);
        \draw[->] (E) -- node[below right,inner sep=0.5mm] {$\psi$} (M);
      \end{tikzpicture}
    \end{center}
  \end{enumerate}
\end{proposition}

\begin{proof}
  \ref{lifti} $\implies$ \ref{liftii}: We define $\psi_c(\sig(w)) =
  c\gam(w)$. Condition \ref{lifti} says that $\psi_c: Q \to M_c$ is
  well-defined.  It is a \homo because $\gamma$ and the left-shift
  $c{\cdot} : Mc \cup \smallset{1} \to M_c, x \mapsto cx$ are \homos
  and $Q\sm\smallset{1} \sse \sig(\Gam^*c)$.

  \ref{liftii} $\implies$ \ref{liftiii}: For $u \in P \sse E(\rho)$ we
  define $\psi (u) = \gam(u) = u \in P \sse M$.  All other elements in
  $E(\rho)$ have the form $(u, \sig(\alp),v)$ with $u,v \in P$ and
  \mbox{$\alp \in \big((P\setminus \smallset{1})^* c\big)^*$}.  Define
  $\psi(u,\sig(\alp),v) = u\psi_c(\sig(\alp))v$. This in an element in
  $M$ because $M_c \sse M$. Now, $\psi_c(\sig(\alp))= c
  \gam(\alp)$. Hence, $\psi(u,\sig(\alp),v) = \gam(uc\alp v)$. Since
  $\gam$, $\tau$ are \homo{}s and $\tau$ is surjective, $\psi$ is a
  \homo, too.
 
  \ref{liftiii} $\implies$ \ref{lifti}: Consider $w \in (P^*c)^*$.  We
  have $\tau(cw) = (1, \sig(w), 1)$. By \ref{liftiii} we have $\gam(cw) 
  = \psi(1, \sig(w), 1) $ . In particular, $\sig(w) = \sig( w')\in
  Q$ implies $c\gam(w) = c\gam( w')\in M_c$.
\end{proof}


\newcommand{\Ju}{Ju}\newcommand{\Ph}{Ph}\newcommand{\Th}{Th}\newcommand{\Ch}{Ch}\newcommand{\Yu}{Yu}\newcommand{\Zh}{Zh}

\end{document}